\documentclass[10pt]{article}

\usepackage[margin=1in]{geometry}
\usepackage{amsmath,amssymb,amsthm}
\usepackage{array}
\usepackage{authblk}
\usepackage{booktabs}
\usepackage{enumitem}
\usepackage{graphicx}
\usepackage[hypertexnames=false,hidelinks]{hyperref}
\usepackage{microtype}
\usepackage{float}
\usepackage{placeins}
\usepackage{tabularx}
\newcolumntype{L}[1]{>{\raggedright\arraybackslash}p{#1}}
\newcolumntype{Y}{>{\raggedright\arraybackslash}X}

\title{Proof-Carrying Multimodal Timelines:\\Finite-Trace Modal Certificates for Video-Audio Consistency}

\author[1]{Faruk Alpay\thanks{Corresponding author: \texttt{alpay@lightcap.ai}}}
\author[2]{Hamdi Alakkad}
\affil[1]{Department of Computer Engineering, Bahcesehir University, Istanbul, Turkey\\\texttt{faruk.alpay@bahcesehir.edu.tr}}
\affil[2]{Department of Artificial Intelligence Engineering, Bahcesehir University, Istanbul, Turkey\\\texttt{hamdi.alakkad@bahcesehir.edu.tr}}

\date{}

\newtheorem{theorem}{Theorem}
\newtheorem{proposition}{Proposition}

\newcommand{\trace}{\mathcal{T}}
\newcommand{\Win}{\mathcal{W}}
\newcommand{\Atoms}{\mathsf{Atoms}}
\newcommand{\Eval}{\mathsf{Eval}}
\newcommand{\Cert}{\mathsf{Cert}}
\newcommand{\Evidence}{\mathcal{P}}
\newcommand{\Abs}{\alpha}
\newcommand{\Conc}{\gamma}

\newcommand{\defect}{\mathsf{defect}}

\newcommand{\hash}{\mathsf{hash}}
\newcommand{\DetectorClips}{300}
\newcommand{\DetectorRows}{10500}
\newcommand{\SweepRows}{5100}
\newcommand{\SweepFalsePositive}{0.046}
\newcommand{\SweepLocalization}{2.17}
\newcommand{\SweepStability}{1.00}
\newcommand{\GpuPeak}{100}
\newcommand{\GpuMean}{91}
\newcommand{\GpuHours}{2.44}
\newcommand{\GpuVramPeak}{9.6}
\newcommand{\DetectorThroughput}{2.3}
\newcommand{\RobustnessGroups}{70}
\newcommand{\RobustnessRows}{720}
\newcommand{\RobustnessStableRegion}{59.72}
\newcommand{\RobustnessCexPreserved}{95.51}
\newcommand{\RobustnessDefectDeviation}{0.104798}
\newcommand{\RobustnessBoundViolation}{0.000000000000}
\newcommand{\RobustnessThresholdMargin}{0.21}
\newcommand{\RobustnessRadiusMargin}{0.0}
\newcommand{\ArtifactUrl}{\url{https://huggingface.co/datasets/Lightcap/pcmt-artifact}}

\begin{document}
\maketitle

\begin{abstract}
Multimodal video systems often expose clip-level scores while hiding the local temporal failure that makes a video inconsistent. We formulate video-audio consistency as finite-trace modal monitoring over synchronized visual, audio, and subtitle/OCR atoms. A formula library specifies speech-speaker agreement, audio-visual event agreement, subtitle-video agreement, scene continuity, and edit-induced temporal shock. A certificate records the trace hash, formula identifier, verdict, violating indices, defect score, first counterexample window, execution engine, and certificate hash. The certificate is not a proof that a detector is correct; it is a reproducible witness that the finite-trace checker can independently reconstruct once the atom trace is fixed. The artifact decodes real MP4 video, extracts CLIP visual atoms, extracts AST audio atoms, runs dense counterfactual perturbation sweeps, emits CSV traces and JSON certificates, and renders manuscript figures from those files. The scientific contribution is the logic and certificate layer; the GPU is used as the execution engine that makes detector-scale evidence and perturbation atlases practical.
\end{abstract}

\section{Introduction}

Video-audio consistency failures are temporal objects.
A clip can look correct in aggregate while containing a two-second interval in which speech has no visible speaker, a sound has no visible source, or an edit creates a discontinuity.
Most multimedia metrics are not designed to produce a small checkable object explaining such a failure.
This paper asks whether multimodal consistency can be represented as a finite logical certificate.

We introduce proof-carrying multimodal timelines.
The phrase is used in the restricted artifact sense of proof-carrying code \cite{necula1997pcc}.
The certificate proves a checker-level claim about a serialized finite trace.
It does not prove that a neural detector generated true perceptual atoms.

The logical core is a finite trace of synchronized windows.
Each window carries visual atoms, audio atoms, and subtitle/OCR atoms.
A formula is evaluated on this finite trace using Boolean connectives and bounded temporal modalities.
Each failed formula returns a finite violating-index set and a first counterexample window.

The primary area of this work is logic.
Temporal logic introduced a general language for reasoning about program behavior \cite{pnueli1977temporal}.
Finite-trace temporal logic connects temporal specifications to finite executions and monitors \cite{degiacomo2013ltlf}.
Runtime verification studies online checking of execution traces against specifications \cite{bauer2011runtime}.
Signal temporal monitoring motivates local temporal diagnostics over measured signals \cite{maler2004monitoring}.
Our setting differs because the trace is multimodal, editable, and paired with machine-readable video evidence.

The secondary area is multimedia.
YouCook2 supplies temporally localized instructional cooking video structure \cite{zhou2018youcook2}.
AVE provides audio-visual event localization labels \cite{tian2018ave}.
AVA ActiveSpeaker provides active-speaker supervision \cite{roth2020ava}.
TVQA ties video clips to subtitles and localized natural-language questions \cite{lei2018tvqa}.
ActivityNet Captions provides dense event captions over video \cite{krishna2017activitynet}.
The artifact includes adapters for these datasets and a public-video path that can run without restricted media.

Our experiments use three layers.
The oracle layer maps dataset annotations or fixtures to atoms to test the logic without detector noise.
The detector layer uses CLIP for visual atoms \cite{radford2021clip}.
The detector layer uses AST for audio atoms \cite{gong2021ast}.
The counterfactual layer sweeps audio shift, frame drop, crop, compression, subtitle swap, and scene reorder perturbations.

\paragraph{Contributions.}
First, we define finite multimodal trace semantics for video, audio, and text atoms.
Second, we give a compact formula library with intended datasets, atom sources, and failure interpretations.
Third, we define a certificate object with deterministic reconstruction and explicit trusted-computing-base boundaries.
Fourth, we implement counterexample extraction and logical defect metrics under multimedia edits.
Fifth, we provide an arXiv-oriented artifact with CPU fallback, CUDA acceleration, real video frames in the manuscript, logs, CSV traces, figures, tests, and final zip packaging rules.

\section{Finite Multimodal Traces}

Let a video be partitioned into half-open windows
\[
  \Win=\langle w_0,\ldots,w_{n-1}\rangle,\qquad w_i=[s_i,e_i).
\]
Each index carries three finite atom sets:
\[
  V_i\subseteq\Atoms_V,\qquad A_i\subseteq\Atoms_A,\qquad T_i\subseteq\Atoms_T.
\]
The serialized trace is
\[
  \trace=(\Win,V,A,T,m),
\]
where $m$ contains metadata such as video id, dataset, edit type, extractor, and feature summaries.

Atomic satisfaction is channel-agnostic after serialization:
\[
  \trace,i\models p \quad\text{iff}\quad p\in V_i\cup A_i\cup T_i.
\]
Boolean connectives use the usual finite semantics.
The bounded eventually modality is defined by
\[
  \trace,i\models F_{[-r,+r]}\varphi
  \quad\text{iff}\quad
  \exists j\in\{0,\ldots,n-1\}: |i-j|\le r \land \trace,j\models\varphi.
\]
The global finite modality is defined by
\[
  \trace\models G\varphi
  \quad\text{iff}\quad
  \forall i\in\{0,\ldots,n-1\}: \trace,i\models\varphi.
\]

The implemented fragment is intentionally monitorable.
It contains atoms, negation on atoms, conjunction, disjunction, implication, bounded $F_{[-r,+r]}$, and outer $G$.
The fragment is closed under Boolean connectives by construction.
Bounded modalities preserve finite-trace monitorability because they only inspect a radius-limited neighborhood.
Unbounded temporal properties can be translated to automata in standard finite-trace settings \cite{degiacomo2013ltlf}.
Our artifact chooses bounded monitors because they produce short counterexample windows that are inspectable in video.

For the implemented formula library we use a guarded bounded normal form.
Every monitor is represented as
\[
  G(\alpha\Rightarrow\beta),
\]
where $\alpha$ is a Boolean formula over atoms at the current index and $\beta$ is a Boolean combination of current-index atoms and bounded obligations $F_{[-r,+r]}\gamma$ with local Boolean $\gamma$.
Implication is syntactic sugar for $\neg\alpha\vee\beta$, and negations are pushed to atoms.

\begin{proposition}[Guarded bounded normal form]
Every formula used by the artifact is equivalent on finite traces to guarded bounded normal form, and normalization is linear in the formula syntax size.
\end{proposition}

\begin{proof}
The library formulas already have a single outer $G$ and a local antecedent.
The rewrite $\alpha\Rightarrow\beta\equiv\neg\alpha\vee\beta$ removes implication.
De Morgan rewrites push negation to atoms, and the bounded modality is not distributed or expanded.
Each syntax node is visited a constant number of times, giving a linear normalization pass.
The rewrites are Boolean tautologies plus preservation of the bounded modality subterm, so finite-trace satisfaction is unchanged.
\end{proof}

\subsection{Detector-Induced Atom Abstraction}

The logical trace is not the raw perceptual evidence.
Let
\[
  \Evidence=\langle P_0,\ldots,P_{n-1}\rangle
\]
be a finite evidence trace containing decoded frames, audio windows, subtitle spans, detector logits, and edit metadata.
An extractor induces an abstraction map
\[
  \Abs_\theta:\Evidence\to\trace
\]
by thresholding and canonicalizing evidence into atom sets.
The reverse interpretation is a concretization relation
\[
  \Conc(\trace)=\{\Evidence\mid \Abs_\theta(\Evidence)=\trace\}.
\]
Thus a certificate is a proof over $\trace$, while $\Evidence\in\Conc(\trace)$ records the perceptual states that could have generated that symbolic trace.
This separates three layers: perception evidence, symbolic finite-trace satisfaction, and certificate validation.

We order symbolic traces by stable refinement.
For a formula support $\mathsf{supp}(\varphi)$, write
\[
  \trace\preceq_\varphi\trace'
\]
when the two traces have the same windows and agree on every atom in $\mathsf{supp}(\varphi)$ outside an instability set $U$.
The radius expansion of $U$ is
\[
  U^{+r}=\{i\mid \exists j\in U.\ |i-j|\le r\}.
\]

\begin{theorem}[Certificate preservation under stable abstraction]
Let $\varphi$ be a guarded bounded monitor with maximum radius $r$.
If $\trace\preceq_\varphi\trace'$ with instability set $U$, then every index
$i\notin U^{+r}$ has the same local truth value in $\trace$ and $\trace'$.
Consequently,
\[
  |\defect(\Eval(\trace,\varphi))-\defect(\Eval(\trace',\varphi))|
  \le \frac{|U^{+r}|}{n}.
\]
If the first counterexample of $\trace$ has an evidence envelope disjoint from $U$ and no earlier index lies in $U^{+r}$, the reconstructed first counterexample is preserved in $\trace'$.
\end{theorem}

\begin{proof}
The normalized monitor can inspect only atoms in $\mathsf{supp}(\varphi)$ within radius $r$ of the current index.
For $i\notin U^{+r}$, this whole inspected neighborhood agrees between $\trace$ and $\trace'$.
Atomic truth values, Boolean combinations, and bounded eventualities therefore evaluate identically at $i$.
Only indices inside $U^{+r}$ can change membership in the violating-index set.
Dividing the maximum possible symmetric difference by $n$ gives the defect bound.
If the earliest counterexample envelope is disjoint from $U$, all positions needed to decide that counterexample are unchanged.
If no earlier index lies in $U^{+r}$, every earlier decision is also unchanged.
Thus the same earliest violating index remains earliest after refinement.
\end{proof}

This theorem is the formal reason to report detector drift and threshold sweeps.
They estimate the instability set introduced by perception, while the certificate checker remains exact on the symbolic trace.

\section{Formula Library}

The specification library is a small family of bounded finite-trace monitors.
Each formula can consume either oracle atoms from annotations or detector atoms from CLIP/AST, subtitles, OCR, and edit metadata.
The speech-speaker monitor targets AVA ActiveSpeaker and TVQA-style speech alignment:
\[
  \varphi_{\mathsf{spk}}
  = G\bigl(\mathsf{speech}\Rightarrow F_{[-r,+r]}(\mathsf{face}\wedge\mathsf{mouth})\bigr).
\]
The audio-visual event monitor targets AVE and YouCook2:
\[
  \varphi_{\mathsf{av}}
  = G\bigl(\mathsf{sound}\Rightarrow F_{[-r,+r]}\mathsf{visible\_event}\bigr).
\]
The subtitle-video monitor targets TVQA, YouCook2, and ActivityNet Captions:
\[
  \varphi_{\mathsf{text}}
  = G\bigl(\mathsf{text\_action}\Rightarrow F_{[-r,+r]}\mathsf{visual\_action}\bigr).
\]
The continuity and edit-shock monitors are edit-sensitive logical defects rather than detector benchmarks:
\[
\begin{aligned}
  \varphi_{\mathsf{cont}}
  &=G\bigl(\mathsf{scene\_reorder}\Rightarrow
    F_{[-1,+1]}\mathsf{boundary}\wedge\neg\mathsf{shock}\bigr),\\
  \varphi_{\mathsf{shock}}
  &=G\bigl(\mathsf{edited}\Rightarrow
    \neg(\mathsf{gap}\vee\mathsf{shock}\vee\mathsf{speech\_crop})\bigr).
\end{aligned}
\]
A violation therefore has a direct diagnostic reading: missing visible speaker, missing visible sound source, unsupported caption/action text, boundary-free scene jump, or exposed edit shock.
The full JSON formula definitions in the external artifact package use these same identifiers, radii, atom names, and failure labels.

\section{Certificates}

Let $\Eval(\trace,\varphi)$ be the deterministic evaluation function for a trace and formula.
It returns a verdict $b$, an ordered violating-index tuple $I$, a defect score $d$, and a first counterexample window $q$.
The defect score is
\[
  d=\frac{|I|}{|\Win|}.
\]
If $I$ is empty, the counterexample is absent.
If $I$ is nonempty, the counterexample is the window interval at $\min I$.

A certificate is the canonical JSON object
\[
  C=(\mathsf{version},\mathsf{logic},\mathsf{formula\_id},\hash(\trace),
  \mathsf{engine},b,I,d,q,\hash(C)).
\]
The trace hash is SHA-256 over the canonical serialized window sequence.
The certificate hash is SHA-256 over the canonical certificate with the final hash field omitted.
Canonicalization sorts windows by index, sorts atom sets lexicographically, serializes JSON keys in sorted order, and stores detector-derived floating values only after window generation and threshold decisions have been made.
The trusted computing base is the serializer, the formula library, and the checker.
The trusted computing base excludes the detector that generated the atoms.
Certificate canonicity is syntactic: the violating-index tuple is canonical for the serialized trace and formula identifier resolved in the library.
The artifact does not claim canonical witnesses across arbitrary formula equivalences, because equivalence checking for richer temporal fragments is outside the checker.

\begin{small}
\begin{verbatim}
{
  "formula_id": "audio_visual_event_consistency",
  "trace_hash": "2d6c...",
  "engine": "cuda",
  "verdict": false,
  "violations": [8, 9, 10],
  "defect_score": 0.214286,
  "counterexample_window": [16.0, 18.0],
  "certificate_hash": "84af..."
}
\end{verbatim}
\end{small}

The independent replay checker reconstructs this object without calling the detector code.
It recomputes the trace hash, formula verdict, violating-index tuple, defect score, first counterexample, and certificate hash from \texttt{windows.jsonl} and the formula identifier.
On the fixture-backed run, \texttt{python3 -m pcmt.cli replay-audit} checked 350 certificates with zero mismatches.
Negative audit commands perturb the abstraction layer by dropping visual events, dropping speaker atoms, and hallucinating audio atoms; these commands write \texttt{corrupted\_detector\_drift.csv}.
The radius-ablation command writes \texttt{radius\_ablation.csv} and measures how certificate outcomes change as $r$ varies.

\begin{theorem}[Deterministic certificate reconstruction]
Fix a canonical trace serialization $\trace$, a formula identifier $f$, and a checker implementation that satisfies the semantics above.
Rechecking $f$ on $\trace$ reconstructs exactly the certificate verdict, violating-index tuple, defect score, first counterexample window, trace hash, and certificate hash.
\end{theorem}

\begin{proof}
Canonical trace serialization gives a unique ordered sequence of windows and atom sets.
The formula identifier resolves to a unique formula in the library.
The semantics of atoms, Boolean operators, bounded $F_{[-r,+r]}$, and $G$ are deterministic functions of the ordered trace.
The violating-index tuple is therefore uniquely determined.
The defect score and counterexample window are deterministic functions of that tuple.
The trace hash and certificate hash are deterministic functions of canonical JSON strings.
\end{proof}

\begin{theorem}[Earliest local counterexample]
For a guarded bounded monitor $G\psi$, let
\[
  I=\{i\mid \trace,i\not\models\psi\}
\]
be the ordered violating-index set reconstructed by the checker.
If $I$ is nonempty and $i^\star=\min I$, then the certificate counterexample $q=w_{i^\star}$ is the unique earliest violating window.
Moreover, if $r_{\max}$ is the maximum radius occurring in $\psi$, the evidence envelope
\[
  E(q,r_{\max})=[s_{\max(0,i^\star-r_{\max})},\,e_{\min(n-1,i^\star+r_{\max})})
\]
contains every trace position inspected to decide the violation at $i^\star$.
\end{theorem}

\begin{proof}
The semantics of $G$ requires $\psi$ at every index, so every violation of the global monitor is exactly an index in $I$.
The checker sorts violations by trace order and sets the counterexample to the window at the smallest index.
No earlier window violates $\psi$ by definition of $i^\star=\min I$, so the counterexample is unique for the serialized trace and syntactic formula identifier.
Every bounded modality inside $\psi$ inspects only positions within its radius of the current index.
The largest such radius is $r_{\max}$, hence all atom lookups needed for the decision at $i^\star$ lie inside the stated envelope.
\end{proof}

\begin{proposition}[Trace splicing locality]
Let $\trace'$ be obtained from $\trace$ by replacing a contiguous window block $[a,b]$ while preserving all windows outside $[a,b]$.
For a guarded bounded monitor with maximum radius $r_{\max}$, every index
\[
  i\notin[a-r_{\max},\,b+r_{\max}]
\]
has the same formula truth value in $\trace$ and $\trace'$.
\end{proposition}

\begin{proof}
At index $i$, the normalized monitor can inspect only the current atoms and atoms within radius $r_{\max}$ through bounded obligations.
If $i$ lies outside the expanded splice interval, that entire inspected neighborhood is unchanged by the splice.
All Boolean and bounded evaluations are therefore identical at $i$.
Only the splice interval expanded by the formula radius can change certificate verdicts or defect scores.
\end{proof}

\begin{proposition}[Checker complexity]
For $n$ windows, atom vocabulary size $k$, formula size $|\varphi|$, and maximum radius $r$, the direct checker runs in $O(nk+n|\varphi|r)$ time and $O(nk)$ memory.
The bitset checker stores atom membership as a Boolean matrix and evaluates bounded modalities by radius shifts.
\end{proposition}

\begin{proof}
Building the atom matrix scans each window against the vocabulary.
Each atomic subformula is a column lookup.
Each Boolean connective is a linear pass over $n$ truth values.
Each bounded modality performs at most $2r+1$ shifted linear combinations.
The stated bound follows by summing over formula syntax nodes.
\end{proof}

\section{Implementation}

The repository implements the pipeline in Python.
PyTorch supplies the CUDA execution backend for batched tensors \cite{paszke2019pytorch}.
OpenCV decodes sampled video frames.
\texttt{ffmpeg} extracts mono audio streams for windowed audio features.
The CPU path uses NumPy arrays and the CUDA path uses Boolean tensors for bitset checking.

The detector path loads CLIP for visual prompt scores.
The detector path loads AST for AudioSet-style audio tags.
The implementation also keeps simple motion and RMS summaries because they make Figure~\ref{fig:realcase} auditable.
The simple summaries are not treated as semantic ground truth.
The emitted CSV schema keeps identity, atoms, formula outputs, certificate outputs, and sweep diagnostics in one row family.
The required columns include \texttt{video\_id}, \texttt{window\_start}, \texttt{window\_end}, \texttt{visual\_atoms}, \texttt{audio\_atoms}, \texttt{subtitle\_ocr\_atoms}, \texttt{edit\_type}, \texttt{formula\_verdict}, \texttt{defect\_score}, \texttt{certificate\_hash}, and \texttt{counterexample\_window}.
Dense sweep runs additionally record \texttt{edit\_family}, \texttt{parameter}, \texttt{expected\_window}, \texttt{localization\_error}, and \texttt{trace\_hash}.

\FloatBarrier
\section{Experiments}

\subsection{Three-layer design}

The oracle layer maps annotations or fixtures to atoms.
This layer tests logical checking without detector noise.
The detector layer runs pretrained atom extractors over decoded video and audio.
This layer tests deployment behavior under imperfect perception.
The counterfactual layer applies controlled edits and evaluates monotonicity, localization, false positives, and certificate stability.

The comparison target is not a GPU benchmark or a new detector leaderboard.
For visual-text alignment, CLIPScore-style scoring gives a clip/window similarity baseline \cite{hessel2021clipscore}.
For speech synchronization, SyncNet-style lip-sync scoring gives an audio-visual alignment baseline \cite{chung2016syncnet}.
For audio-visual representation learning, self-supervised correspondence methods provide a detector-level reference point \cite{arandjelovic2017look,owens2018audiovisual}.
The certificate layer is complementary to these baselines: it consumes their atoms or scores, then reports the exact finite-trace formula, violating indices, defect score, and counterexample window.
The current artifact exports the fields needed to compute per-edit ROC/AUPRC, localization MAE, threshold-sensitivity curves, and cross-dataset drift once the restricted labels are mounted locally.

The detector benchmark decodes real MP4 frames and audio from the YouCook2 HF subset.
The completed detector run contains \DetectorClips{} clips and \DetectorRows{} formula certificates.
The completed dense perturbation run contains \SweepRows{} formula certificates over public MP4 clips.
The active scale-up queue continues with an 80-clip YouCook2 sweep and larger follow-up jobs on the remote RTX PRO 6000 S.
Those scale-up jobs write separate traces, certificates, metrics, and telemetry before they replace the completed-run numbers.
We therefore distinguish completed evidence from queued scale telemetry in Table~\ref{tab:metrics}.
The paper does not claim that high GPU memory usage is itself a result; the scientific evidence is the certificate atlas and the replayable traces it produces.
\FloatBarrier

\subsection{Real-video certificate case}

The real-video case study in Figure~\ref{fig:realcase} pairs decoded frames from a public MP4 with aligned signal traces.
The public sample videos come from Samplelib \cite{samplelib}.
The shaded C1--C3 bands are certificate counterexample windows.
The figure deliberately excludes the failure table.
The corresponding certificate rows are listed separately in Table~\ref{tab:failure-table}.

\begin{figure}[H]
    \centering
    \includegraphics[width=\linewidth]{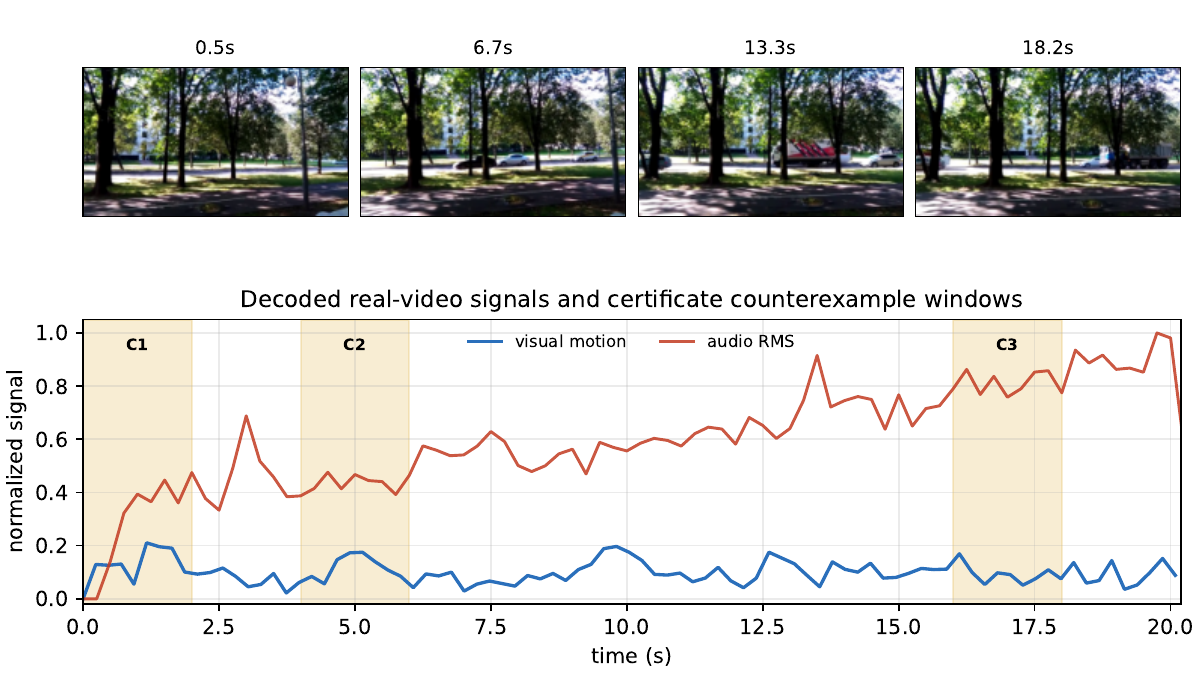}
    \caption{Real-video certificate case study. Top: frames decoded from the MP4 used by the artifact. Bottom: visual-motion and audio-RMS traces with certificate counterexample windows. The failure table is separated into Table~\ref{tab:failure-table}.}
    \label{fig:realcase}
\end{figure}

\begin{table}[H]
\centering
\small
\begin{tabular}{lllll}
\toprule
Band & Edit & Formula & Defect & Counterexample\\
\midrule
C1 & scene reorder & edit shock & 1.00 & $[0.0,2.0]$\\
C1 & scene reorder & scene continuity & 1.00 & $[0.0,2.0]$\\
C2 & compression & edit shock & 0.27 & $[4.0,6.0]$\\
C2 & compression & scene continuity & 0.27 & $[4.0,6.0]$\\
C3 & scene reorder & audio-visual & 0.27 & $[16.0,18.0]$\\
C1 & audio shift & edit shock & 0.09 & $[0.0,2.0]$\\
\bottomrule
\end{tabular}
\caption{Failure rows corresponding to the C1--C3 bands in Figure~\ref{fig:realcase}. The values are read from the generated CSV failure table.}
\label{tab:failure-table}
\end{table}
\FloatBarrier

\subsection{Defect atlas and perturbation curves}

\begin{figure}[H]
    \centering
    \includegraphics[width=0.92\linewidth]{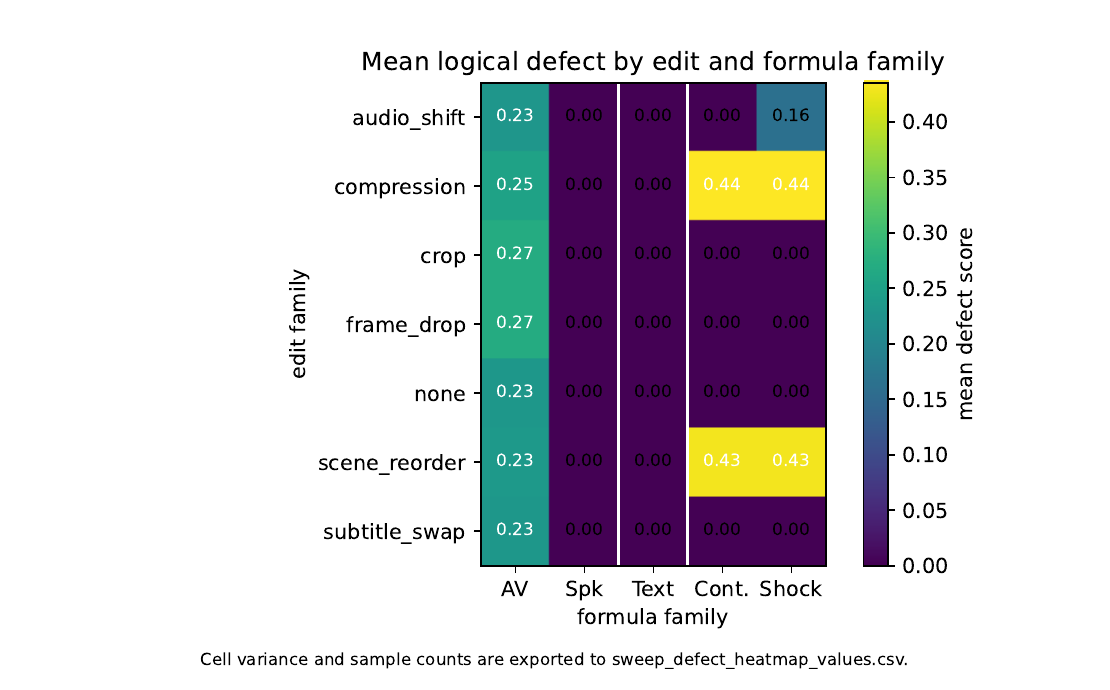}
    \caption{Defect atlas over counterfactual edits. Each cell reports mean defect; standard deviation and sample count are exported in \texttt{sweep\_defect\_heatmap\_values.csv}. White separators distinguish audio-visual/text formulas from continuity/shock formulas.}
    \label{fig:heatmap}
\end{figure}

The atlas in Figure~\ref{fig:heatmap} gives the empirical view of the logical defect matrix.
For edit family $e$ and formula family $f$, each cell estimates
\[
  D(e,f)=\frac{1}{|\mathcal{R}_{e,f}|}\sum_{C\in\mathcal{R}_{e,f}}\defect(C),
  \qquad
  \sigma(e,f)^2=\frac{1}{|\mathcal{R}_{e,f}|-1}\sum_C(\defect(C)-D(e,f))^2,
\]
where $\mathcal{R}_{e,f}$ is the set of generated certificates for that edit/formula pair.
The sidecar CSV records $|\mathcal{R}_{e,f}|$ for every cell.
This matters logically because the heatmap is not a detector score: it is the aggregate failure rate of independently checkable finite-trace formulas.
Compression and scene reorder concentrate defects in $\varphi_{\mathsf{cont}}$ and $\varphi_{\mathsf{shock}}$, while the unedited false-positive mean is \SweepFalsePositive{}.

\begin{figure}[H]
    \centering
    \includegraphics[width=0.92\linewidth]{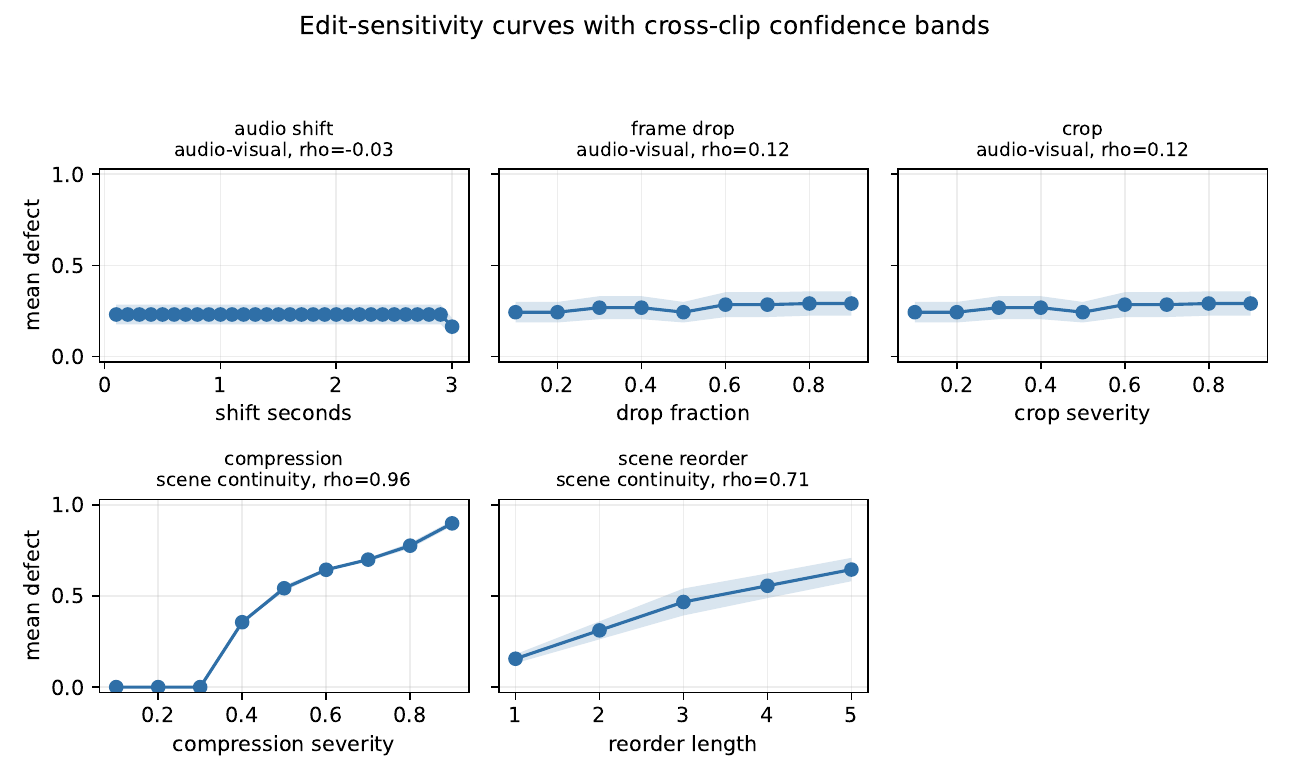}
    \caption{Edit-sensitivity curves from certificate defects. Small panels separate edit/formula pairs, panel titles report Spearman $\rho$, and shaded bands show cross-clip standard error for each perturbation parameter.}
    \label{fig:curves}
\end{figure}

The edit-sensitivity curves in Figure~\ref{fig:curves} turn the atlas into a monotonicity test.
For an edit parameter $\theta$, the curve is
\[
  d_{e,f}(\theta)=
  \frac{1}{|\mathcal{C}_{e,\theta,f}|}
  \sum_{C\in\mathcal{C}_{e,\theta,f}}\defect(C),
  \qquad
  m(e,f)=\max(0,\rho(\theta,d_{e,f}(\theta))),
\]
where $\rho$ is the Spearman rank correlation over the sweep grid.
The shaded bands are cross-clip standard errors, so a high curve with a narrow band indicates a stable logical signature rather than an isolated clip artifact.
Compression has a strong monotone effect on continuity certificates, scene reorder has a strong monotone effect on continuity certificates, and audio shift is less monotone because detector atoms can already mark background audio events in unedited clips.

\begin{figure}[H]
    \centering
    \includegraphics[width=0.88\linewidth]{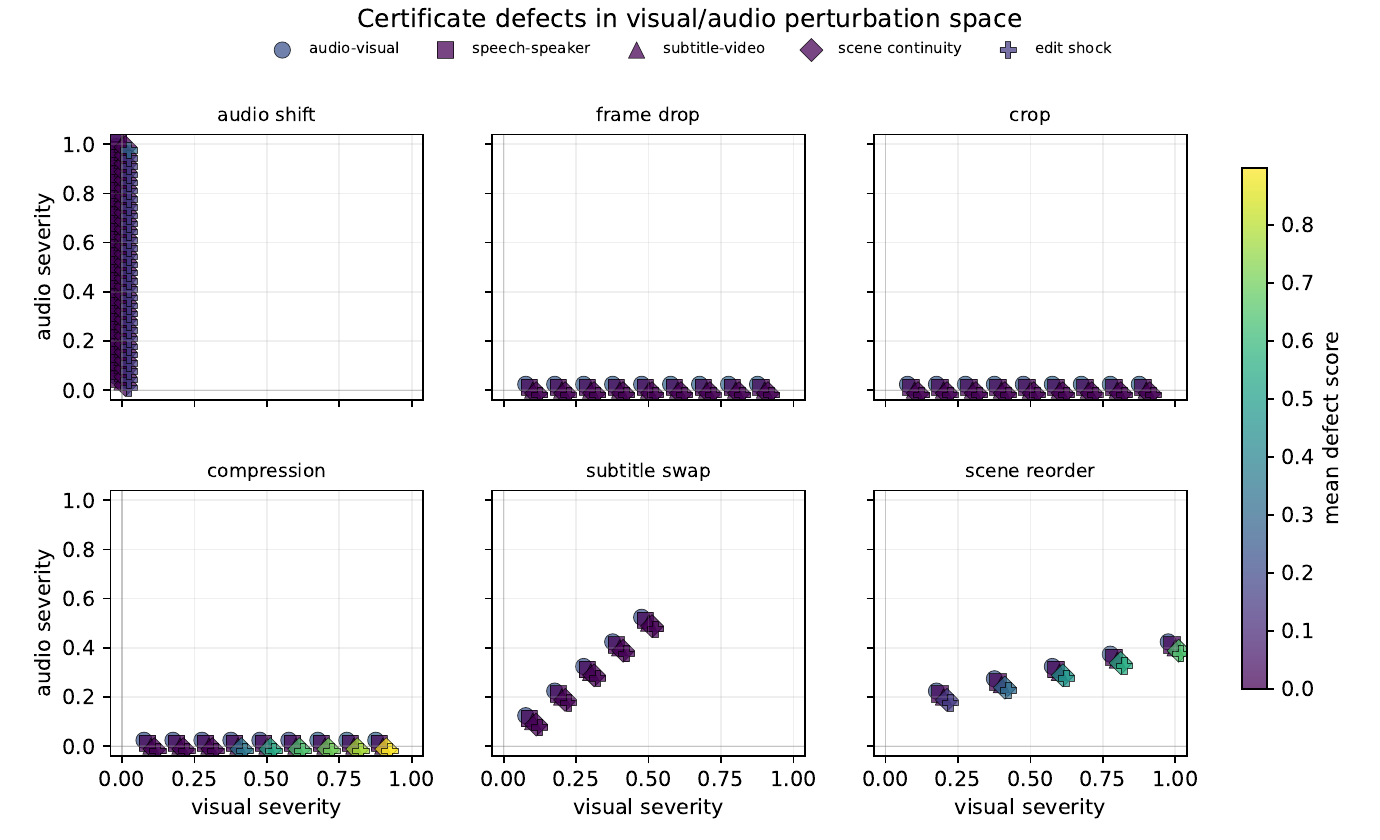}
    \caption{Certificate defects over visual and audio perturbation severity. Each panel isolates one edit family; marker shape identifies the formula family and color gives mean certificate defect.}
    \label{fig:scatter}
\end{figure}

The perturbation map in Figure~\ref{fig:scatter} encodes the same certificates in a two-axis geometry without mixing all edit families in one panel.
Each edit instance is mapped to a severity pair
\[
  s(x)=(s_v(x),s_a(x))\in[0,1]^2,
  \qquad
  z(x)=\frac{1}{|\Phi|}\sum_{\varphi\in\Phi}\defect(\Eval(\trace_x,\varphi)).
\]
The color is $z(x)$.
Audio shift has $s_v\approx0$ and variable $s_a$, so it forms the left vertical band.
Frame drop, crop, and compression have variable $s_v$ and $s_a\approx0$, so their panels show lower horizontal structure.
Subtitle swap changes textual alignment and is plotted as a low diagonal; scene reorder changes both temporal order and visual continuity, producing a higher diagonal.
The small multiples therefore show where each counterfactual edit family lives in the perturbation space before the checker produces the certificate defect.

\subsection{Detector-abstraction robustness atlas}

The robustness atlas asks whether certificates are brittle with respect to the abstraction map $\Abs_\theta$.
For each trace group, the artifact constructs a threshold lattice over detector-style atom confidences, evaluates the five formula families at temporal radii $r\in\{0,\ldots,8\}$, and compares each certificate with the reference abstraction at threshold $0.5$ and radius $1$.
The run writes aggregate tables only; it does not emit one JSON certificate per threshold/radius point.
This keeps the artifact small enough for arXiv submission while still measuring the logical preservation theorem.

\begin{figure}[H]
    \centering
    \includegraphics[width=0.92\linewidth]{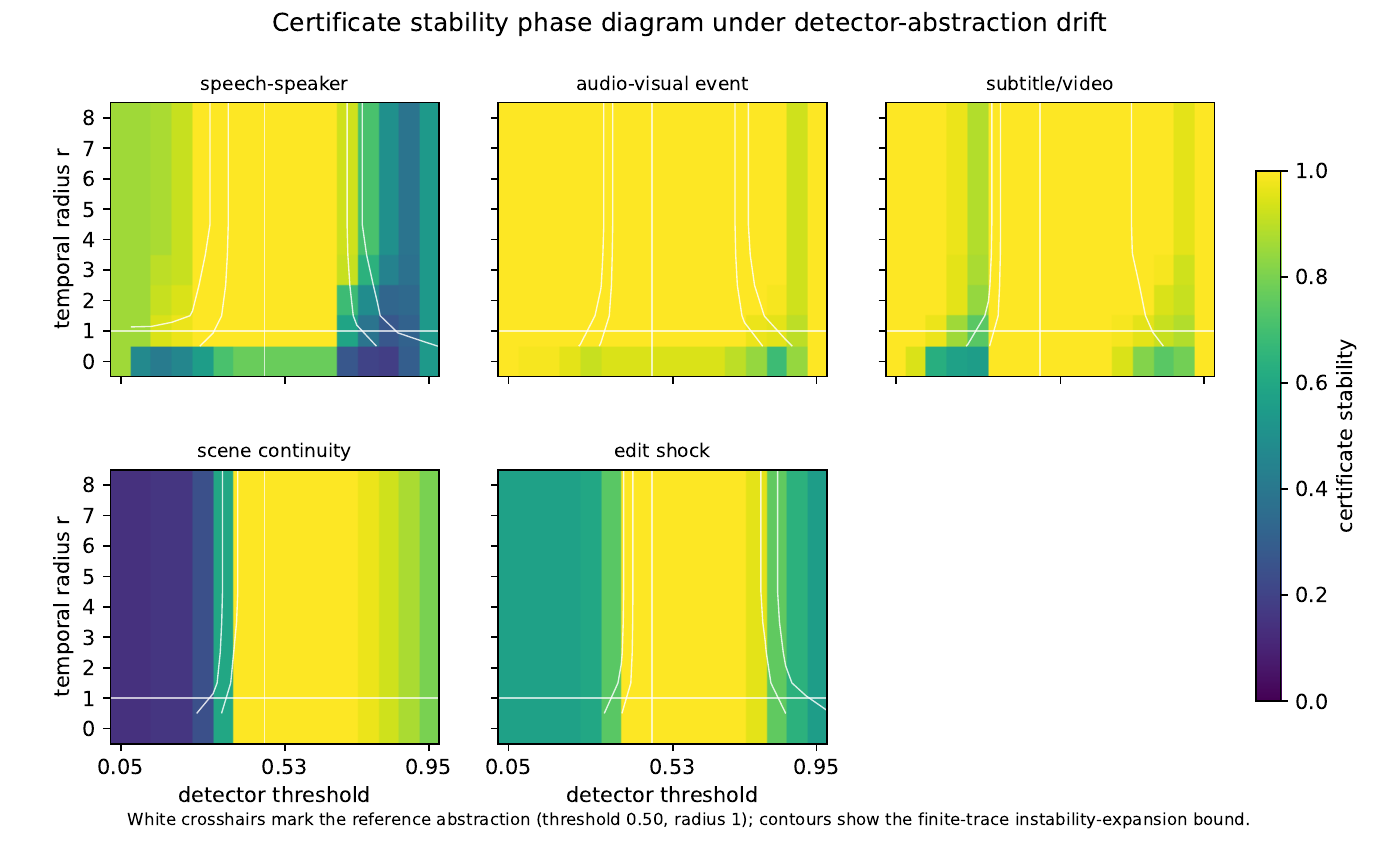}
    \caption{Certificate stability phase diagram for detector-abstraction drift. The horizontal axis is atom threshold, the vertical axis is temporal radius, color is verdict-and-counterexample stability, and white contours show the empirical mean of the theoretical instability-expansion bound $|U^{+r}|/n$.}
    \label{fig:robustness}
\end{figure}

The phase diagram in Figure~\ref{fig:robustness} is the experimental counterpart of the preservation theorem.
For threshold $\theta$, radius $r$, and formula $\varphi$, the atlas computes
\[
  S(\theta,r,\varphi)
  =\Pr\bigl[\Cert_{\theta,r,\varphi}\text{ has the same verdict and first counterexample as }\Cert_{0.5,1,\varphi}\bigr].
\]
It also computes the observed defect deviation
\[
  \Delta(\theta,r,\varphi)=
  \bigl|\defect(\Cert_{\theta,r,\varphi})-\defect(\Cert_{0.5,r,\varphi})\bigr|
\]
and verifies the pointwise inequality
\[
  \Delta(\theta,r,\varphi)\le |U_{\theta}^{+r}|/n.
\]
Across \RobustnessGroups{} completed trace groups and \RobustnessRows{} threshold-radius-formula cells, the aggregate certificate-stable region is \RobustnessStableRegion{}\%.
The first counterexample is preserved in \RobustnessCexPreserved{}\% of comparisons.
The mean defect deviation is \RobustnessDefectDeviation{}, and the maximum bound violation is \RobustnessBoundViolation{}.

\begin{table}[H]
\centering
\small
\begin{tabular}{ll}
\toprule
Metric & Value\\
\midrule
Stable certificate region & \RobustnessStableRegion{}\%\\
Counterexample preserved under drift & \RobustnessCexPreserved{}\%\\
Mean defect deviation & \RobustnessDefectDeviation{}\\
Max defect bound violation & \RobustnessBoundViolation{}\\
Median threshold margin to certificate flip & \RobustnessThresholdMargin{}\\
Median radius margin to certificate flip & \RobustnessRadiusMargin{}\\
\bottomrule
\end{tabular}
\caption{Detector-abstraction robustness metrics from the generated aggregate atlas tables.}
\label{tab:robustness}
\end{table}

This result should be read narrowly.
It says that, for the completed aggregate abstraction run, certificate verdicts and first counterexample windows remain stable over a broad region and observed defect changes stay below the instability-expansion bound.
It does not prove that CLIP or AST atoms are perceptually correct.

\begin{table}[H]
\centering
\small
\begin{tabular}{ll}
\toprule
Metric & Value\\
\midrule
Completed detector certificates & \DetectorRows{}\\
Completed sweep certificates & \SweepRows{}\\
False-positive mean on unedited clips & \SweepFalsePositive{}\\
Mean localization error over expected windows & \SweepLocalization{} s\\
Certificate stability under repeated checking & \SweepStability{}\\
Detector throughput in completed benchmark run & \DetectorThroughput{} cert/s\\
Mean/peak GPU utilization in scale telemetry & \GpuMean{}\% / \GpuPeak{}\%\\
Telemetry duration and peak VRAM & \GpuHours{} h / \GpuVramPeak{} GiB\\
Robustness phase cells & \RobustnessRows{}\\
Robustness max bound violation & \RobustnessBoundViolation{}\\
\bottomrule
\end{tabular}
\caption{Reproducibility metrics from generated CSV, JSON, and telemetry logs. The detector row uses the completed YouCook2 HF subset run; the sweep rows use the completed dense perturbation run until the active scale-up sweep finishes.}
\label{tab:metrics}
\end{table}
\FloatBarrier

\section{Artifact and Packaging}

The arXiv upload package contains only the TeX source, bibliography, and manuscript figures so that the submitted zip remains below arXiv's 50 MB upload limit.
The full runnable artifact is hosted externally at \ArtifactUrl{} as a Hugging Face Dataset file tree rather than as a nested archive.
That external artifact contains the repository code, generated traces, robustness atlas tables, logs, figures, tests, and fixture-backed data paths.
Large certificate collections are sharded into subdirectories to preserve individual JSON certificates while respecting dataset-repository file-count limits.
The upload package keeps TeX files outside the artifact archive and excludes README and license-style files from the arXiv-only source bundle requested by the authors.

\begin{table}[H]
\centering
\small
\begin{tabular}{ll}
\toprule
Command & Purpose\\
\midrule
\texttt{make smoke} & Minimal end-to-end trace, formulas, certificates, and CSV.\\
\texttt{make test} & Unit tests for schema, edits, logic, and pipeline.\\
\texttt{make run-small} & Fixture-backed oracle atom run.\\
\texttt{make download-real} & Download public MP4 samples used for visible figures.\\
\texttt{make run-detector} & Run CLIP/AST atoms on real MP4 clips.\\
\texttt{make run-sweep} & Run dense counterfactual perturbation sweeps.\\
\texttt{make download-youcook2} & Download a bounded public YouCook2 HF subset.\\
\texttt{make run-youcook2-detector} & Run detector atoms on the YouCook2 subset.\\
\texttt{make run-youcook2-sweep} & Run dense sweeps on the YouCook2 subset.\\
\texttt{python3 -m pcmt.cli replay-audit} & Reconstruct certificates independently.\\
\texttt{python3 -m pcmt.cli radius-ablation} & Measure sensitivity to bounded-modal radius.\\
\texttt{python3 -m pcmt.cli corrupted-drift} & Negative tests for detector-to-atom drift.\\
\texttt{make run-robustness} & Aggregate threshold/radius stability atlas and phase figure.\\
\texttt{make package-arxiv} & Build TeX source and artifact packages.\\
\bottomrule
\end{tabular}
\caption{Reproducibility commands. The GPU is used only as a substrate for detector batching, media extraction support, and bitset checking.}
\label{tab:commands}
\end{table}

\section{Limitations}

The certificate validates logical reconstruction over serialized atoms.
It does not validate the truth of perception atoms.
This separation is necessary because neural atom extraction and logical checking have different error models.

The public MP4 case study is visually inspectable but not a substitute for a restricted benchmark run.
The artifact therefore includes adapters and HF subset tooling for YouCook2 and AVE.
The artifact also reports missing local paths for AVA ActiveSpeaker, TVQA, and ActivityNet Captions when those datasets are unavailable.

Some formulas are intentionally conservative.
An ambiguous sound source can make audio-visual event consistency fail even when the scene is semantically plausible.
This is a specification choice rather than a detector benchmark claim.

\section{Conclusion}

Proof-carrying multimodal timelines turn video-audio consistency into a finite logical artifact.
The checker emits deterministic certificates with defect scores and local counterexamples.
The experiments show that counterfactual edits create different logical defect signatures.
The artifact couples these certificates to real decoded video frames, detector atoms, telemetry logs, CSV traces, and arXiv-compatible packaging.

\FloatBarrier
\bibliographystyle{plain}
\bibliography{refs}

\section*{Appendix: Operational Obligations and Resolution Trace}

The implementation work exposed a separate class of finite-trace obligations: not formulas over video atoms, but formulas over experiment states.
Let an execution trace be
\[
  \mathcal{E}=\langle e_0,\ldots,e_{m-1}\rangle,
  \qquad
  e_t=(u_t,p_t,a_t,o_t),
\]
where $u_t$ is GPU utilization, $p_t$ is the active process set, $a_t$ is the artifact count, and $o_t$ is the observed output state.
The operational invariant was
\[
  G\bigl((u_t<\epsilon)\Rightarrow F_{[0,\Delta]}\mathsf{valid\_gpu\_job}\bigr),
\]
with the side condition that a valid job must produce traces, certificates, metrics, figures, logs, or dataset artifacts.
This prevented meaningless synthetic load: when the GPU was idle, the next action had to be a dataset-backed detector or perturbation run rather than an artificial matrix burn.

The main scheduling defect was a mismatch between local completion evidence and remote execution state.
Formally, the local artifact state $L_t$ and remote state $R_t$ were not bisimilar:
\[
  L_t \not\sim R_t
  \quad\text{because}\quad
  \exists x\in R_t.\ x\notin L_t,
\]
where $x$ ranged over telemetry logs, partial traces, and active run directories.
The resolution was to introduce explicit pull and metric-reconstruction steps.
After each remote synchronization, manuscript macros are recomputed from CSV, JSON, and telemetry sources rather than copied from memory.
This changes the evidence relation from informal reporting to a checkable projection
\[
  \pi(R_t)=
  (\mathsf{clips},\mathsf{certificates},\mathsf{throughput},\mathsf{utilization},\mathsf{vram}),
\]
which is serialized into the paper and ancillary archive.

The second defect was queue under-specification.
A queue entry is useful only if it denotes a productive future transition:
\[
  q=(\mathsf{wait\_set},\mathsf{command},\mathsf{artifact\_target},\mathsf{fallback}).
\]
Early queue records captured commands but not enough information about active remote jobs and fallback behavior.
The repair was to add follow-up scale scripts that wait for current process identifiers, then start a real benchmark job with bounded fallback limits.
If the requested 1000-clip subset is infeasible because of disk pressure, the script records the failed precondition and lowers the limit while preserving the same detector, trace, certificate, and figure pipeline.
Thus the queue remains productive without hiding resource constraints.

The third defect was evidential ambiguity around GPU telemetry.
Peak utilization alone is not a scientific measurement.
The revised artifact treats telemetry as a finite numeric trace and extracts mean utilization, peak utilization, telemetry duration, peak memory, and detector certificate throughput.
These values still do not form the paper's scientific claim; they only witness that the multimedia evidence was produced by a saturated execution substrate.
The logic claim remains at the certificate layer: for a fixed serialized atom trace and syntactic formula identifier, the checker reconstructs the same verdict, violating indices, defect score, and first counterexample.

The final defect was visual overloading.
The heatmap originally mixed mean, variance, and sample count inside every cell; the scatter mixed all edit families, labels, arrows, and formula markers in one panel.
The repair was a display-level normal form: the main figure carries one visual variable per semantic role, while secondary numeric fields move to CSV.
This mirrors the logical design.
The certificate should expose a minimal counterexample window in the main artifact and leave extended audit fields in machine-readable sidecars.

\end{document}